\documentclass{article}

\usepackage[ruled]{algorithm2e} 
\usepackage{booktabs} 
\usepackage[square,numbers]{natbib}
\usepackage{authblk}
\usepackage{fullpage}
\usepackage{amsmath, amssymb, amsthm}
\usepackage{tikz}
\usetikzlibrary{shapes,arrows,fit,calc,positioning}
\usetikzlibrary{arrows}
\usetikzlibrary{decorations.markings}
\usepackage{enumerate}
\usepackage{hyperref}
\usepackage{graphicx}
\usepackage{caption}
\usepackage{hyperref}
\usepackage{subcaption}
\usepackage{mathtools}
\usepackage{comment}
\usepackage[mathscr]{eucal}
\SetAlFnt{\small}
\SetAlCapFnt{\small}
\SetAlCapNameFnt{\small}
\SetAlCapHSkip{0pt}
\IncMargin{-\parindent}
\usepackage{thm-restate}
\usepackage{subcaption}
\usepackage{array}
\usepackage{float}
\usepackage{amsmath}
\usepackage{comment}

\usepackage{xcolor}
\usepackage{hyperref}

\usepackage{float}
\theoremstyle{definition}
\newtheorem{theorem}{Theorem}

\newtheorem{lemma}[theorem]{Lemma}

\newtheorem{definition}[theorem]{Definition}
\newtheorem{observation}[theorem]{Observation}

\newtheorem{corollary}[theorem]{Corollary}
\newtheorem{example}[theorem]{Example}
\newtheorem{mclaim}[theorem]{Claim}
\theoremstyle{remark}

\begin{document}
\title{Almost Envy-freeness, Envy-rank, and Nash Social Welfare Matchings}

\author[1]{Alireza Farhadi}
\author[1]{MohammadTaghi Hajiaghayi}
\author[2]{Mohamad Latifian}
\author[3]{Masoud Seddighin}
\author[1]{Hadi Yami}
\affil[1]{University of Maryland, College Park}
\affil[2]{Sharif University of Technology}
\affil[3]{Institute for Research in Fundamental Sciences (IPM)}

\maketitle              

\newcommand*{\QEDA}{\hfill\ensuremath{\blacksquare}}%
\newcommand*{\QEDB}{\hfill\ensuremath{\square}}%

\newcommand{\OPT}{\ensuremath{\mathsf{OPT}}\xspace}
\newcommand{\PC}{\textsf{P}\xspace}
\newcommand{\NP}{\textsf{NP}\xspace}

\newcommand{\agents}{\mathcal{N}}
\newcommand{\items}{\mathcal{M}}
\newcommand{\ite}{b}
\newcommand{\allocation}{\mathcal{A}}
\newcommand{\valu}{v}
\newcommand{\valus}{V}
\newcommand{\EF}{\mathsf{EF}}
\newcommand{\EFX}{\mathsf{EFX}}
\newcommand{\EFO}{\mathsf{EF1}}
\newcommand{\EFR}{\mathsf{EFR}}
\newcommand{\NSW}{\mathsf{NSW}}
\newcommand{\ef}{\mathsf{EF}}
\newcommand{\efx}{\mathsf{EFX}}
\newcommand{\efo}{\mathsf{EF1}}
\newcommand{\efr}{\mathsf{EFR}}
\newcommand{\nsw}{\mathsf{NSW}}
\newcommand{\group}{\mathsf{G}}
\newcommand{\distribution}{\mathsf{D}}
\newcommand{\MMS}{\mathsf{MMS}}

\newcommand{\agent}{a}

\newcommand{\prob}[1]{\textsc{#1}\xspace}
\newcommand{\algo}[1]{\textsc{#1}\xspace}
\newcommand{\length}{\text{length}\xspace}
\newcommand{\Mandate}{\text{Mandate}}
\newcommand{\len}{\length}
\newcommand{\weight}{\text{weight}\xspace}
\newcommand{\weightenclosed}{\text{weight-enclosed}}
\newcommand{\cost}{\text{cost}}
\newcommand{\g}{\textsf{sol}}
\newcommand{\h}{\textsf{except}}
\newcommand{\parent}{\ensuremath{p}}
\newcommand{\mincut}{\textsf{mincut}\xspace}
\newcommand{\dist}{\text{\rm dist}}
\newcommand{\B}{\ensuremath{\mathcal{B}}}
\newcommand{\set}[1]{\{#1\}}
\newcommand{\calF}{\tt Enclosed}
\newcommand{\structured}{structured}
\newcommand{\nearOPT}{\mbox{nearOPT}}
\newcommand{\OPTsolution}{\ensuremath{\hat L}}
\newcommand{\F}{\ensuremath{\mathcal{F}}}
\newcommand{\pmwcut}{\prob{Multiway Cut}}
\newcommand{\pstcut}{\prob{Minimum $s\text{-}t$ Cut}}
\newcommand{\nph}{\textsf{NP-hard}\xspace}
\newcommand{\apxh}{\textsf{APX-hard}\xspace}
\newcommand{\ptas}{\textsf{PTAS}\xspace}
\newcommand{\psolve}{\textsf{P}\xspace}
\newcommand{\npsolve}{\textsf{NP}\xspace}
\newcommand{\eps}{\ensuremath{\epsilon}}
\newcommand{\opt}{\OPT}
\newcommand{\dt}[1]{\emph{#1}}
\newcommand{\MG}{M}
\newcommand{\PRG}{PCG}
\newcommand{\brickdecomp}{\algo{Brick-Decomposition}}
\newcommand{\mwcutspan}{\algo{MainSpanner}}
\newcommand{\mwcutskeleton}{\algo{Skeleton}}
\newcommand{\mwcutspanner}{\algo{Spanner}}
\newcommand{\findear}{\algo{Find\-Ear}}
\newcommand{\parkone}{\algo{Find\-Short\-Cycle}}
\newcommand{\parktwo}{\algo{Find\-Short\-Cycle\-With\-Range}}
\newcommand{\parkthree}{\algo{Find\-Short\-Separating\-Cycle}}
\newcommand{\encl}{\ensuremath{\mathsf{enclosed}}}
\newcommand{\union}{\ensuremath{\cup}}
\newcommand{\intersect}{\ensuremath{\cap}}
\newcommand{\pclose}{p_{close}}
\newcommand{\pfar}{p_{far}}
\newcommand{\E}{\mathop{\mathlarger{\mathbb{E}}}}
\newcommand{\of}{\bar{s}}
\newcommand{\on}{s}
\newcommand{\Mod}{\ \mathrm{mod}\ }

\definecolor{mygreen}{RGB}{20,140,80}
\definecolor{mylightgray}{RGB}{230,230,230}

\definecolor{mygreen}{RGB}{20,140,80}
\definecolor{mydarkgray}{gray}{0.15} 
\definecolor{oceanblue}{HTML}{2c55c2}
\hypersetup{
     colorlinks=true,
     citecolor= mygreen,
     linkcolor= black
}

\newcommand{\ca}[1] {\textcolor{oceanblue}{ #1 }}

\newcommand{\citeboth}[1]{\hypersetup{citecolor=mydarkgray}\citeauthor{#1}\hypersetup{citecolor=mygreen} \cite{#1}}

\newcommand\mycommfont[1]{\textcolor{oceanblue}{#1}}
\SetCommentSty{mycommfont}



\begin{abstract}
Envy-free up to one good ($\efo$) and envy-free up to any good ($\efx$) are two well-known extensions of envy-freeness for the case of indivisible items. It is shown that $\efo$ can always be guaranteed for agents with subadditive valuations \cite{lipton2004approximately}. In sharp contrast, it is unknown whether or not an $\efx$ allocation always exists, even for four agents and additive valuations. In addition, the best approximation guarantee for $\efx$ is $(\phi -1) \simeq 0.61$ by Amanitidis et al. \cite{amanatidis2019multiple}.

In order to find a middle ground to bridge this gap, in this paper 
we  suggest another fairness criterion, namely \emph{envy-freeness up to a random good} or $\efr$, which is weaker than $\efx$, yet stronger than $\ef1$. For this notion, we provide a polynomial-time $0.73$-approximation allocation algorithm. 
For our algorithm we introduce Nash Social Welfare Matching which makes a connection between Nash Social Welfare and envy freeness. We believe Nash Social Welfare Matching will find its applications in future work.

\end{abstract}

\section{Introduction}
Fair division is a fundamental and interdisciplinary problem that has been extensively studied in economics, mathematics, political science, and computer science
\cite{farhadi2019fair,seddighin2019externalities,ghodsi2018fair,dickerson2014computational,aziz2018fair,kurokawa2018fair,plaut2018almost,lipton2004approximately,nash1950bargaining,barman2017approximation,chaudhury2019little,amanatidis2019multiple,caragiannis2019unreasonable}. Generally, the goal is to find an allocation of a resource to $n$ agents, which is agreeable to all the agents according to their preferences. The first formal treatment of this problem was in 1948 by 
Steinhaus  \cite{Steinhaus:first}. Following his work,  a vast literature has been developed and several notions for measuring fairness have been suggested \cite{Steinhaus:first,Foley:first,Budish:first,lipton2004approximately,caragiannis2019unreasonable}.  One of the most prominent and well-established fairness notions, introduced by Foley \cite{Foley:first}, is envy-freeness,  which requires that each agent prefers his  share over that of any other agent.

Traditionally, envy-freeness has been studied for both divisible and indivisible resources. When the resource is a single heterogeneous divisible item (i.e, can be fractionally allocated), envy-freeness admits strong theoretical guarantees. For example, it is shown that allocations exist that allocate the entire resource, and are both envy-free and Pareto efficient\footnote{
	An allocation is Pareto efficient if it is not possible to reallocate the resources such that at least one agent is better off without making any other person worse off.
} and allocate each agent a contiguous piece of the resource \cite{varian1973equity}. Apart from mere existence, there are algorithms that find an envy-free allocation for arbitrary number of agents \cite{aziz2016discrete,brams,dehghani2018envy}. However, beyond divisibility, when dealing with a set of indivisible goods, envy-freeness is too strong to be attained; for example, for two agents and a single indivisible good, the agent that receives no good envies another party. Therefore, several relaxations of envy-freeness are introduced for the case of indivisible items  \cite{lipton2004approximately,Budish:first,caragiannis2019unreasonable}. One of these relaxations,  suggested by Budish \cite{Budish:first}, is \textit{envy-freeness up to one good} ($\EFO$)\footnote{
	It is worth to mention that before the work of Budish \cite{Budish:first}  $\EFO$ was implicitly addressed by Lipton et. al \cite{lipton2004approximately}. 
}. An allocation of indivisible goods is $\EFO$ if any possible envy of an agent for the share of another agent can be resolved by removing some good from the envied share. In contrast to envy-freeness, $\EFO$ allocation always exists. Indeed, a simple round-robin algorithm always guarantees $\EFO$ for additive valuations, and a standard envy-graph based allocation guarantees $\EF1$ for more general (sub-additive) valuations. Besides, it is shown that any Nash welfare maximizing allocation (allocation that maximizes the product of the agents' utilities) is both Pareto efficient and $\EFO$.

Recently, Caragiannis \textit{et al.} \cite{caragiannis2019unreasonable} suggested
another intriguing relaxation of envy-freeness, namely \textit{envy-free up to any good}  ($\EFX$), which attracted a lot of attention. An allocation is said to be $\EFX$, if no agent envies another agent after the removal of any item from the other agent's bundle. Theoretically, this notion is strictly stronger than $\EF1$ and is strictly weaker than $\EF$. In contrast to $\EF1$,  questions related to $\EFX$ notion is relatively unexplored. As an example, despite significant effort \cite{caragiannis2019unreasonable}, the existence of such allocations is still unknown. The most impressive breakthrough in this  area is the recent work of Chaudhury, Garg, and Mehlhorn \cite{chaudhury2020efx}, which shows that for the case of $3$ agents with additive valuations $\EFX$ allocation always exists.  Furthermore, unlike $\efo$, Nash social welfare maximizing allocations are not necessarily $\efx$ \cite{caragiannis2019unreasonable}. 

Given this impenetrability of $\efx$, a growing strand of research started considering its relaxations. For example, Plaut and Roughgarden \cite{plaut2018almost}, consider an approximate version of $\efx$\footnote{
	An allocation is $\alpha$-approximate $\EFX$, if for every pair of agents $i$ and $j$, agent $i$ believes that the share allocated to him is worth at least $\alpha$ fraction of the share allocated to agent $j$, after removal of agent $j$'s least valued item (according to agent $i$'s preference). 
}
and provide a $1/2$ approximation solution for agents with sub-additive valuation functions. For additive valuations, this factor is recently improved to $0.618$ by Amanatidis \textit{et al.} \cite{amanatidis2019multiple}.  Another interesting relaxation is \emph{$\efx$-with-charity}. Such allocations donate a bundle of  items to charity and guarantee $\efx$ for the rest of the items. The less valuable the donated items are, the more desirable the allocation is.  Caragiannis \textit{et al.}\cite{caragiannis2019envy}  show that there always exists an $\EFX$-with-charity allocation where every agent receives half the value of his bundle in the optimal Nash social welfare allocation. Recently, Chaudhury \textit{et al.} \cite{chaudhury2019little} have proposed an $\EFX$-with-charity allocation such that no agent values the donated items more than his bundle and the number of donated items is less than the number of agents. 

Considering the huge discrepancy between $\efx$ and $\efo$, in this paper we wish to  find a middle ground to bridge this gap. 
We therefore suggest another fairness criterion, namely \emph{envy-freeness up to a random item} or $\efr$, which is weaker than $\efx$, yet stronger than $\ef1$. For this notion, we provide a polynomial time 0.73-approximation algorithm, i.e., an algorithm that constructs 0.73-$\efr$ allocations in polynomial time. Our allocation method is based on a special type of matching, namely Nash Social Welfare Matching. In Section \ref{results}, we briefly discuss our techniques to obtain these results.

\subsection{Our Results and Techniques} \label{results}


\paragraph{\textbf{Envy-freeness up to a random item.}}
We suggest a new fairness notion, namely \emph{evny-free up to a random good} ($\efr$). Roughly speaking, in an $\efr$ allocation, no agent $i$ envies another agent $j$ (in expectation), if we remove a random good from the bundle of agent $j$. In other words, the expected value of agent $i$ for the bundle allocated to agent $j$, after removing a random item from it is at most as much as the value of his own bundle. 
Obviously,
$\efr$ is a weaker notion than $\efx$, yet stronger than $\efo$.

	The intuition behind $\EFR$ is to use randomness to reduce the \emph{severe impact of small items}. To see what we mean by this term, consider the following scenario: suppose that the value of agent $i$ for his share is $1000$. In addition, assume that the bundle allocated to an agent $j$ contains two items, each with value $600$ to agent $i$. Even though the allocation is currently $\efx$ with respect to agent $i$,  allocating even a very small item (say, with value close to $0$ to agent $i$) to agent $j$ violates $\efx$ condition for agent $i$. This is counter-intuitive in the sense that the last item allocated to agent $j$ was totally worthless to agent $i$. On the other hand, allocating any item with value less than $300$ to agent $j$ preserves $\efr$ condition for agent $i$. This property makes $\efr$ more flexible, especially when the number of items is not too much. 
On the other hand, as the number of items allocated to an agent grows larger, we expect $\efx$ and $\efr$ to be more and more aligned. 

Similar to $\efx$, we provide a counter example which shows that a Nash Social Welfare allocation is not necessarily $\efr$ (see Example \ref{nsw_example}). This separates $\efr$ from $\efo$ given the fact that a Nash Social welfare allocation is always $\EF1$\cite{caragiannis2019unreasonable}.  It is worth mentioning that  Caragiannis et al. \cite{caragiannis2019unreasonable}  presented an example to show that Nash Social welfare allocation is not necessarily $\EFX$. However, their example is still $\EFR$. The difference between these two examples can be seen as  an evidence for the distinction between $\EFR$ and $\EFX$. 

As noted, the best approximation guarantee for $\efx$ is $0.61$ by Amanatidis \textit{et al.} \cite{amanatidis2019multiple}. Since every $\efx$ allocation is also $\efr$, this result also provides a $0.61$-approximation algorithm for $\efr$. In this paper, we improve this ratio to $0.73$. 

\begin{restatable}{theorem}{efrthm}
	\label{main2}
	There exists an algorithm that finds a $0.73$-$\efr$ allocation. In addition, such an allocation can be found in polynomial time. 
\end{restatable}

In order to prove Theorem \ref{main2}, we propose a three-step algorithm that finds a $0.73$-$\efr$ allocation in polynomial time.
Roughly speaking, in the first two steps, we allocate valuable (i.e., large) items while preserving the $0.73$-$\efr$ property. Next, we use an envy-cycle based procedure to allocate the rest of the items. Figure \ref{floch} shows a flowchart of our method.

\begin{figure}[H]
	\centering
	\includegraphics[width=0.5\linewidth]{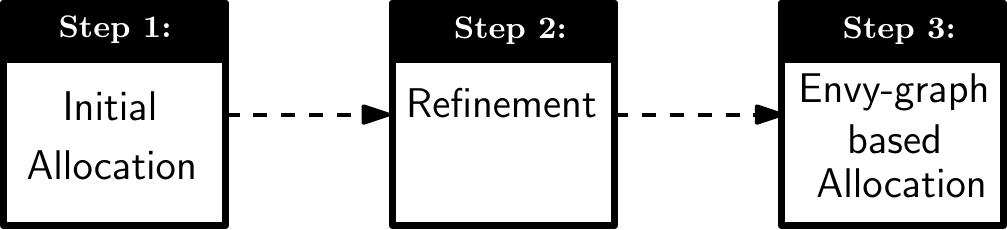}
	\caption{Flowchart of the $0.73$-$\efr$ allocation algorithm}
	\label{floch}
\end{figure}

The first challenge to address is the method by which we must allocate large items in the first step. Interestingly, we introduce a special type of matching allocation with intriguing properties which makes it ideal for our algorithm. We call such an allocation a \emph{Nash Social Welfare Matching}.

\paragraph{\textbf{Nash Social Welfare Matching.}}
In the first step of the algorithm, we allocate one item to each agent such that the product of the utilities of the agents is maximized. The interesting fact about this allocation is that, not only does this allocation allocates large items, but it also provides very useful information about the value of the rest of the items. In Section \ref{prelim} we broadly discuss such allocations and their properties. However, to shed light on their usefulness, assume that after a Nash Social Welfare Matching, agent $i$ envies agent $j$ with a ratio $\alpha>1$, meaning that he thinks the value of the good allocated to agent $j$ is $\alpha$ times more than the value of his item. In that case, we can immediately conclude that the item allocated to agent $j$ is $\alpha$ times more valuable to him (agent $j$) than any remaining item; otherwise, we could improve the utility product by allocating the most valuable  remaining item to agent $j$ and giving his former item to agent $i$ (and of course, freeing agent $i$'s former item). In addition, we can express the same proposition for the value of the item allocated to agent $i$ for agent $j$: the value of this item for agent $j$ is at most $1/\alpha$ of the item allocated to agent $j$. The above statement can be generalized to the arguments that include more than two agents. 
With this aim, we introduce several new concepts, including \emph{envy-ratio graph} (a complete weighted graph that  represents the envy-ratios between agents), \emph{improving cycles}, and \emph{envy-rank}.

It is worth mentioning that the main challenge in many fair allocation problems for different fairness criteria (e.g.,  $\MMS$, $\EFX$) is allocating valuable items. 
The structure of such matchings makes them ideal for allocating these items. We strongly believe that using Nash Social Welfare matching is not only useful for our algorithm, but can also be seen as a strong tool in the way of finding fair allocations related to the other fairness notions, especially maximin-share.  In Section \ref{sec:efx} we show how to use $\NSW$ mathcing to obtain a simple algorithm with the approximation ratio of $(\phi-1) \simeq 0.61$ for $\efx$. The approximation ratio of our algorithm matches the state-of-the-art $(\phi-1)$ approximation result by Amanitidis et al. \cite{Amanatidis2018comparison}.

\subsection{Related work}
Fair allocation of a divisible resource (known as \emph{cake cutting}) was first introduced by Steinhaus\cite{Steinhaus:first} in 1948, and since then has been the subject of intensive studies. We refer the reader to \cite{brams1996fair} and \cite{robertson1998cake} for an overview of different fairness notions and their related results. Proportionality and Envy-freeness are among the most well-established notions for cake cutting. As mentioned, the literature of cake cutting admits strong positive results for these two notions (see \cite{Steinhaus:first} for  details). 

Since neither $\ef$ nor proportionality or any approximation of these notions can be guaranteed  for indivisible goods, several relaxations are introduced for these two notions in the past decade. These relaxations include $\efo$ and $\efx$ for envy-freeness and maximin-share \cite{Budish:first} for proportionality. 
Nash Social Welfare ($\nsw$) is also another important notion in allocation of indivisible goods which is somewhat a trade off between fairness and optimality.

Apart from the results mentioned in the introduction for $\efx$ and $\efo$, there are other studies related to these notions\cite{barman2018finding,barman2018NSW,barman2018greedy,caragiannis2019envy,caragiannis2019unreasonable,chaudhury2018fair}. In particular, Barman \textit{et al.} \cite{barman2018finding}
propose a pseudo-polynomial time algorithm that finds an $\efo$ and pareto efficient allocation. They also show that any $\efo$ and pareto efficient allocation approximates Nash Social Welfare with a factor of $1.45$. 
In contrast to $\EFO$, our knowledge of $\efx$ and $\nsw$ beyond additive valuations is  limited. For $\efx$, the only positive results for general valuations is the work of Plaut and Roughgarden \cite{plaut2018almost} which provides a $1/2$-$\efx$ allocation. For $\nsw$, Grag et al. \cite{garg2020approximating} prove an $O(n\log n)$ approximation  guarantee for submodular valuations. Recently this factor is improved to $O(n)$\cite{chaudhury2020fair}.


Maximin-share is one of the most well-studied notions in the recent years. 
In a pioneering study, Kurokawa \textit{et al.} \cite{kurokawa2018fair} provide  an approximation algorithm with the factor of $2/3$ for maximin-share, which is improved to $3/4$ by Ghodsi et al \cite{ghodsi2018fair}.
Beyond additivity, Barman \textit{et al.} \cite{barman2017approximation} show that a simple round robin algorihtm can guarantee  $1/10$-$\MMS$ for submodular valuations, and Ghodsi \textit{et al.}  provide approximation guarantees for submodular ($1/3$), XOS ($1/5$) and subadditive ($1/\log n$) valuations.
In addition, several notions are ramified from maximin-share, including weighted maximin-share ($\textsf{WMMS}$) \cite{farhadi2019fair}, pairwise maximin-share ($\textsf{PMMS}$)\cite{caragiannis2019unreasonable}, and groupwise maximin-share ($\textsf{GMMS}$)\cite{barman2018groupwise}. Several studies consider the relation between these notions  and seek to find an allocation that guarantees a subset of them simultaneously. For example, Amanatidis \textit{et al.} \cite{Amanatidis2018comparison} investigate the connections between $\efo$, $\EFX$, maximin share, and pairwise maximin share.  They show that any $\EF1$ allocation is also a $1/n$-$\MMS$ and a $1/2$-$\textsf{P}\MMS$ allocation. They also prove that any $\EFX$ allocation is a $4/7$-$\MMS$ and a $2/3$-$\textsf{PMMS}$ allocation.

\section{Preliminaries and Basic Observations}\label{prelim}

\paragraph{\textbf{Fair allocation problem}.} An instance of fair allocation problem consists of a set of $n$ agents, a set $\items$ of $m$ goods, and a valuation profile $\valus = \{\valu_1,\valu_2,\ldots,\valu_n\}$. Each $v_i$ is a function of the  form $2^{\items} \rightarrow \mathbb R_{\geq 0}$ which specifies the preferences of agent $i \in [n]$ over the goods. Throughout the paper, we assume that a valuation function $v_i$ satisfies the following conditions.
\begin{itemize}
	\item \textbf{Normalization}: $\valu_i(\emptyset) = 0$.
	\item \textbf{Monotonicity}: $\valu_i(S) \leq \valu_i(T)$ whenever $S \subseteq T$.
	\item \textbf{Additivity}: $\valu_i(S) = \sum_{\ite \in S} \valu_i(\{\ite\})$.
\end{itemize}

An allocation of a set $S$ of goods is an $n$-partition  $\allocation=\langle \allocation_1,\allocation_2,\ldots,\allocation_n \rangle$ of $S$, where $\allocation_i$ is the bundle allocated to agent $i$. Allocation is complete, if $S = \items$ and is partial otherwise. Since we are interested in the allocations that allocate the whole set of items, the final allocation must be complete. 

\paragraph{\textbf{Fairness critera}.} Given an instance of fair division problem and an allocation $\allocation$, an agent $i$ envies another agent $j$, if he strictly prefers $\allocation_j$ over his bundle $\allocation_i$. An allocation is then said to be \emph{envy-free} ($\EF$), if no agent envies another, i.e., for every pair $i,j \in [n]$ of agents we have $\valu_i(\allocation_i) \geq \valu_i(\allocation_j)$. As mentioned, envy-freeness is too strong to be guaranteed in an allocation of indivisible items. Therefore, two relaxations of this notion are introduced, namely \emph{envy-free up to one good} ($\EF1$) and \emph{envy-free up to any good} ($\EFX$).

\begin{definition}
	An allocation $\allocation$ is called
	\begin{itemize}
		\item envy-free up to one good $(\EF1)$ if for all $i,j$ we have
		$\valu_i(\allocation_i) \geq \min_{\ite \in \allocation_j} \valu_i(\allocation_j \setminus \{\ite\})$,
		\item envy-free up to any good $(\EFX)$ if for all $i,j$ we have 
				$\valu_i(\allocation_i) \geq \max_{b \in \allocation_j} \valu_i(\allocation_j \setminus \{\ite\})$.
	\end{itemize} 
\end{definition}

Even though these two notions seem to be somewhat related, there is a huge discrepancy between the current results obtained for them. 
It is shown that 
even for instances with general valuations, an $\EF1$ allocation always exists, and can be computed in polynomial time \cite{lipton2004approximately}. In contrast, whether or not an $\EFX$ allocation always exists is still open, even for additive valuations. 

In this paper, we introduce another relaxation of envy-freeness, namely \emph{envy-free up to a random good}. Let $\distribution_j$ be a uniform distribution over the items of $\allocation_j$ that selects each item with probability $1/|\allocation_j|$.

\begin{definition}
	Allocation $\allocation$ is envy-free up to a random good ($\EFR$) if for all $i,j$ we have
	$$\valu_i(\allocation_i) \geq \mathop{\mathlarger{\mathbb{E}}}_{\ite \sim \distribution_j} \bigg[\valu_i(\allocation_j \setminus \{\ite\})\bigg] \,.$$
\end{definition}
Clearly, $\EFR$  lies in between $\EFX$ and $\EFO$: $\EFX$ is a stronger notion of fairness than $\EFR$, and $\EFR$ is stronger than $\EF1$.
 In Example \ref{nsw_example}, we show one  structural difference between $\efo$ and $\efr$: in contrast to $\efo$, $\efr$ is not implied by an allocation that maximizes Nash social welfare.  

\begin{example}
		\label{nsw_example}
			\begin{figure}[h]
			\centering			\setlength\extrarowheight{8pt}
			\begin{tabular}{c|p{0.5cm} p{0.5cm} p{0.5cm} p{0.5cm} c}
				&$1$&$2$&$3$&$4$&5\\
				\hline
				$\valu_1$&3&3&1&1&1\\ 
				$\valu_2$&5&5&1&4&3
				
			\end{tabular}
			\caption{Agents' valuations over items}
			\label{nsw_example_table}
		\end{figure}
Consider an instance of the fair allocation problem with 5 items, and 2 agents with the valuations represented in Figure \ref{nsw_example_table}. The unique allocation that maximizes the $\nsw$ allocates the first 3 items to the first agent, and the other 2 items to the second agent. Let $\allocation$ be  this allocation. Since there are 3 items in the first agent's bundle, we have
\begin{align*}
\mathop{\mathlarger{\mathbb{E}}}_{\ite \sim \distribution_1} \bigg[\valu_2(\allocation_1 \setminus \{\ite\})\bigg] \, &= 
\frac{1}{3} \cdot \big (\valu_2(\allocation_1 \setminus \{1\}) + \valu_2(\allocation_1 \setminus \{2\}) + \valu_2(\allocation_1 \setminus \{3\}) \big) \\
&= \frac{22}{3} \geq \valu_2(\allocation_2) = 7 \,.
\end{align*} 
Hence, this allocation is not  $\efr$.

\end{example}

Finally, approximate versions of $\efx$ and $\efr$ are defined as follows.

\begin{definition}
	For a constant $c \le 1$, an allocation $\allocation$ is called
	\begin{itemize}
		\item $c$-approximate envy-free up to any good ($c$-$\EFX$), if for all $i,j$ we have
		$$\valu_i(\allocation_i) \geq c \cdot \max_{\ite \in \allocation_j} \valu_i(\allocation_j \setminus \{\ite\}) \,,$$
		\item $c$-approximate envy-free up to a random good ($c$-$\EFR$) if for all $i,j$ we have 
		$$\valu_i(\allocation_i) \geq c \cdot \mathop{\mathlarger{\mathbb{E}}}_{\ite \sim \distribution_j} \bigg[\valu_i(\allocation_j \setminus \{\ite\})\bigg] \,.$$
	\end{itemize} 
	Note that Example \ref{nsw_example} also shows that the maximum $\nsw$ allocation does not guarantee better than $\frac{21}{22}$ approximation of $\efr$.
\end{definition}
\paragraph{\textbf{Envy-ratio Graph}.}
Envy-ratio graph is in fact a generalization of envy-graph introduced by Lipton et al. \cite{lipton2004approximately}. Suppose that at some stage of our algorithm we have a partial allocation $\allocation$. We define a graph called \emph{envy-ratio graph}  to be a complete weighted digraph with the following construction: each vertex corresponds to an agent, and for each ordered pair $(i,j)$, there is a directed edge from vertex $i$ to vertex $j$ with the weight $w_{i,j} = \valu_i(\allocation_j) / \valu_i(\allocation_i)$. 

Assuming each agent has a non-zero value for each good, for every $i,j$ we have $w_{i,j} \in [0,\infty)$. Note that  $w_{i,j} \le 1$ implies that agent $i$ does not envy agent $j$, whereas $w_{i,j}>1$ indicates agent $i$ envies agent $j$. The higher the value of $w_{i,j}$ is, the more envious agent $i$ is to the bundle of agent $j$. Indeed, the well-known envy-graph is a subgraph of envy-ratio graph containing only the edges with $w_{i,j}>1$.

\paragraph{\textbf{Nash Social Welfare (NSW) Mathcing}. }  

Nash social welfare, originally proposed by Nash \cite{nash1950bargaining}, 
 is defined to be the geometric mean of agents' valuations. Allocation that maximizes Nash social welfare is known to have desirable properties. For example, such allocations are proved to be $\EF1$ and pareto optimal. Roughly, Nash social welfare maximizing allocations can be seen as a trade-off between social welfare and fairness. 
 
 In the first step of the algorithm, we allocate one item to each agent such that the Nash social welfare of the agents is maximized. More formally, define Nash Social Welfare matching of $[m]$ to be a partial allocation $\allocation = \langle \allocation_1, \allocation_2,\ldots, \allocation_n \rangle$, such that  $\Pi_i \valu_i(\allocation_i)$ is maximized and for every $i$ we have $|\allocation_i|=1$.

Similar to Nash social welfare allocations, Nash social welfare matchings exhibit beautiful properties which greatly help us in designing our algorithm. One simple property of such allocations is shown in Observation \ref{obs1}. Before we state Observation \ref{obs1}, we need to define concepts of \emph{improving} and \emph{strictly improving} cycles.

\begin{definition}
	Let $c = i_1 \rightarrow i_2 \rightarrow \ldots \rightarrow i_k \rightarrow i_1$ be a cycle in the envy-ratio graph. Then, $c$ is an improving cycle, if $$w_{i_1,i_2} \times w_{i_2,i_3} \times \ldots \times w_{i_{k-1},i_k} \times w_{i_k,i_1} > 1 \,.$$
	Furthermore, we say a cycle $c$ is  strictly improving cycle, if $c$ is an improving cycle and  for every $(i\rightarrow j) \in c$,  $w_{i,j}>1$ holds. 
\end{definition}

We note that strictly improving cycle is an essential concept  in all envy-cycle elimination methods \cite{lipton2004approximately,barman2017approximation,chaudhury2019little,amanatidis2019multiple}. These methods typically rotate the shares over  strictly improving cycles to enhance social welfare. However, to the best of our knowledge, no previous work made use of improving cycles. 

\begin{restatable}{observation} {obsf}
\label{obs1}
Suppose that we allocate one item to each agent using Nash social welfare matching. Then, the envy-ratio graph admits no improving cycle. 
\end{restatable}

The proof of the mentioned observation is available in Appendix \ref{sec:appen}. A particularly useful case of Observation \ref{obs1} is for the cycles of length $2$, which we state in Corollary \ref{col:vremaining}.

\begin{corollary}[of Observation \ref{obs1}]\label{col:vremaining}
	Suppose that for two agents $i,j$ we have $\valu_i(\allocation_j) \geq r \cdot \valu_i(\allocation_i)$, where $r \ge 1$. Then, we have $\valu_j(\allocation_i) \leq \valu_j(\allocation_j)/r$. 
\end{corollary}

\begin{definition}
	Suppose that we allocate one item to each agent using Nash social welfare matching. We define the envy-rank of an agent $i$, denoted by $r_i$ as 
	\begin{align*}
	r_i = \max_{j_0,j_1,\ldots,j_k} 
	\prod_{z=1}^{k} w_{j_z,j_{z-1}} \,,
	\end{align*}
\end{definition}
where $j_0 =i$.
Roughly speaking, let $p$ be a path leading to vertex $i$ in the envy-ratio graph such that the product of the weights of the edges in $p$ is maximum. Then, the envy-rank of agent $i$ equals to the product of the weights of the edges in $p$. Note that by Observation \ref{obs1} we can assume w.l.o.g that $p$ is a simple path (i.e., $p$ includes  no duplicate vertices).

\begin{observation}
  $p$ is a simple path. 
\end{observation}
\begin{proof}
	Assume $p$ is not simple and let $c$ be a cycle in $p$. By Observation \ref{obs1} we know that $c$ can not be improving.  Therefore, the product of the weight of the edges of $p\setminus c$ is at least as large as $p$.
\end{proof}

To get a better understanding of these definitions take a look at Example \ref{example}.

\begin{example}
	\label{example}
		Consider an instance of the fair allocation problem with 4 items, 4 agents, and a valuation profile $\valus = \{\valu_1, \valu_2, \valu_3, \valu_4\}$ represented in Figure \ref{example_table}. Let $\allocation$ be the allocation that allocates item $i$ to agent $i$. The envy-ratio graph and the envy graph of $\allocation$ are shown in Figure \ref{example_erg} and Figure  \ref{example_eg} respectively. This allocation is not envy-free, however, it is both $\efx$ and $\efr$ since each agent receives only one item.
		
As we mentioned before, the envy-rank of an agent can be seen as the product of the weights of the edges in a path leading to that agent. For instance, consider the agent $1$. The envy-rank of this agent is $3$ which is the product of the weights of the edges in the path $3 \rightarrow 2 \rightarrow 1$. Also consider the cycle $1 \rightarrow 3 \rightarrow 2 \rightarrow 1$. This cycle is an improving cycle. Therefore,  allocation $\allocation$ is not a $\nsw$ matching. The allocation can be improved by moving the items alongside this cycle which leads to a new allocation $\allocation' = \langle \{3\}, \{1\}, \{2\}, \{4\} \rangle$. The envy-ratio graph of $\allocation'$ can be seen in Figure \ref{fig:erg2}.


\begin{figure}[t]
	\begin{subfigure}[b]{0.35\textwidth}
		\centering			\setlength\extrarowheight{4pt}
		\begin{tabular}{c||p{0.35cm} p{0.35cm} p{0.35cm} c}
			&$1$&$2$&$3$&$4$\\
			\hline
			$\valu_1$&8&2&4&3\\ 
			$\valu_2$&4&2&0&2\\ 
			$\valu_3$&0&3&2&2\\
			$\valu_4$&1&6&3&9\\
			
		\end{tabular}
	\vspace{1cm}
		\caption{Agents' valuations over items}
		\label{example_table}
	\end{subfigure}
	\begin{subfigure}[b]{0.45\textwidth}
		\centering
		\includegraphics[width=0.7\textwidth]{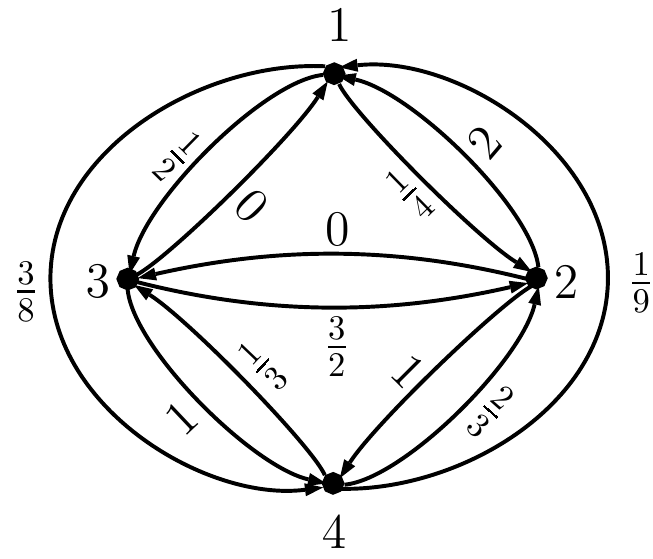}
		\caption{Envy-ratio graph}
		\label{example_erg}
	\end{subfigure}
	\begin{subfigure}[b]{0.35\textwidth}
		\vspace{1.1cm}
		\centering
		\includegraphics[width=0.4\textwidth]{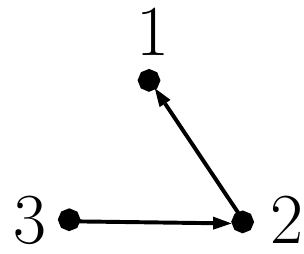}
		\vspace{1.1cm}
		\caption{Envy graph}
		\label{example_eg}
		
	\end{subfigure}
	\begin{subfigure}[b]{0.45\textwidth}
		\centering
		\vspace{0.2cm}
		\includegraphics[width=0.7\textwidth]{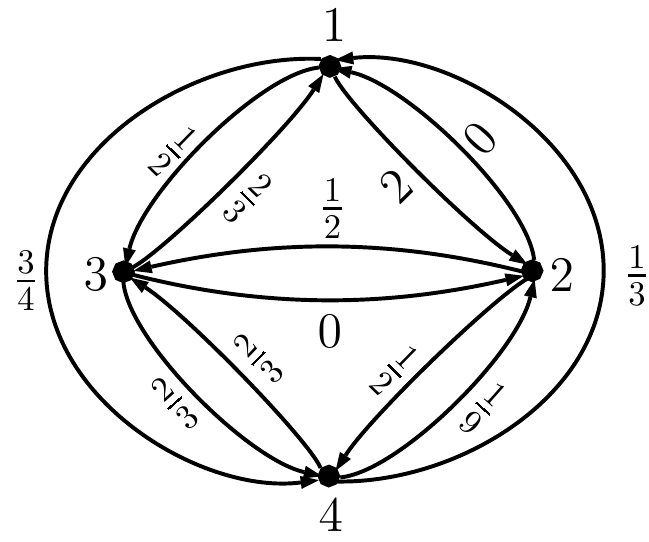}
		\caption{Envy-ratio graph after eliminating a cycle.}
		\label{fig:erg2}
	\end{subfigure} 
	\caption{An example to illustrate envy-ratio graph.}
\end{figure}

\end{example}

We finish our discussion in this section by mentioning some properties of envy-rank values. The proofs of flowing observations are available in Appendix \ref{sec:appen}.

\begin{restatable}{observation} {valloc}
	\label{clm:valloc}
	Suppose that allocation $\allocation$ allocates one item to each agent using a Nash social welfare matching. Then for every pair of agents $i$ and $j$, we have
	\begin{align*}
	\frac{\valu_i(\allocation_j)}{\valu_i(\allocation_i)} \le \min\big\{r_j, \frac{r_j}{r_i}\big\} \,.
	\end{align*}
\end{restatable}

In addition to Observation \ref{obs1}, Nash social welfare matchings admit another important and elegant property, which we state in 
Observation  \ref{obs:vremaining}. This observation provides upper bounds on the value of remaining goods and can be of independent interest for various fair allocation problems. 
\begin{restatable}{observation} {vremaining}
		\label{obs:vremaining}
		Suppose that we allocate one item to each agent using a Nash social welfare matching. Then, for each agent $i$ and any unallocated item $\ite$ we have 
		$$
		\valu_{i}(\ite) \leq \min\big\{ \valu_i(\allocation_i), \frac{\valu_i(\allocation_i)}{r_i} \big\} \,.
		$$
\end{restatable}

%

\section{An approximate EFR Allocation}
\label{sec:efr}
In this section, we present our algorithm for finding a  $0.73$-$\EFR$ allocation.  Our algorithm is divided into 3 steps, namely $\nsw$ matching, allocation refinement, and envy-graph based allocation. In the first step, we allocate each agent one item using a Nash social welfare matching and accordingly divide the agents into three groups based on their envy-rank. Next, in the second step we allocate a set of goods to the agents in each group, and finally in the third step we allocate the rest of the items using the classic envy-cycle elimination method. The outline of our algorithm is represented in Algorithm \ref{algo}.

\SetKwInOut{Parameter}{Parameters}

\begin{algorithm}[H]
\Parameter {$\varphi=\sqrt{3}+1$.}
\SetAlgoLined
\tcp{Step 1}
Allocate $\nsw$ matching\;
Let $r_i$ be envy-rank of an agent $i$. Divide the agents into groups $\group_1$, $\group_2$, $\group_3$ as follows. 	Agent $i$ belongs to $\group_1$ if $r_i > \varphi$, belongs to $\group_2$ if $2<r_i\leq \varphi$, and belongs to $\group_3$ if $r_i \leq 2$\;
\tcp{Step 2}
Let $\mathcal{O}$ be a topological ordering of the agents with respect to the envy-graph\;
\ForEach{$i \in \group_3$ ordered by $\mathcal O$}
{
	Ask agent $i$ to pick his most valuable remaining item\;
}
\ForEach{$i \in \group_3$ ordered by $\mathcal O$}
{
	Ask agent $i$ to pick his most valuable remaining item\;
}
\ForEach{$i \in\group_2$ ordered by $\mathcal O$}
{
	Ask agent $i$ to pick his most valuable remaining item\;
}
\tcp{Step 3}
 \While{the allocation is not complete}{
 Eliminate all directed cycles in the envy-graph\;
 	Let $s$ be an arbitrary source in the envy-graph\;
	Ask agent $s$ to pick his most valuable remaining item\;
 }

\Return the allocation;
 \caption{The outline of the $0.73$-$\efr$ algorithm.}
 \label{algo}
\end{algorithm}

\subsection{Step 1.}
In the first step, we allocate one item to each agent using a $\nsw$ matching. We first show that this allocation can be found in polynomial time. The proof is available in Appendix \ref{sec:appen}.

\begin{restatable}{observation}{poly}
$\nsw$ matching can be found in polynomial time.
\end{restatable}

Let $\allocation$ be $\nsw$ matching and fix a parameter $\varphi=\sqrt{3}+1$. Based on the envy-rank of the agents, we divide them into $3$ groups $\group_1,\group_2,$ and $\group_3$ as follows.
\begin{itemize}
	\item Agent $i$ belongs to $\group_1$ if $r_i > \varphi$. 
	\item Agent $i$ belongs to $\group_2$ if $2<r_i\leq \varphi$.	
	\item Agent $i$ belongs to $\group_3$ if $r_i \leq 2$.
	\end{itemize}
Note that by Observation \ref{obs:vremaining}, we know that for every remaining item $b$ the following properties hold.
\begin{itemize}
	\item (Property 1): For every agent $i \in \group_1$ we have $\valu_{i}(b) < \valu_{i}(\allocation_i)/\varphi$. 
	\item (Property 2): For every agent $i \in \group_2$ we have $\valu_{i}(b) < \valu_{i}(\allocation_i)/2$. 
\end{itemize}

Intuitively, if each remaining item is worth less than $\valu_i(\allocation_i) /\varphi$ to every agent $i$, then we can guarantee the approximation factor of $1/(1+1/\varphi)$ in the third step. This property holds for the agents in $\group_1$; however, 
 this is not the case for agents in $\group_2$ and $\group_3$. In the second step, we seek to allocate a set of items to the agents in $\group_2$ and $\group_3$ so that the same property holds for these agents. Note that alongside this property, the final partial allocation after the second step must be fair (i.e., $0.73$-$\efr$).

\subsection{Step 2.}
In the second step, we allocate one item to each agent in $\group_2$ and two items to each agent in $\group_3$. Algorithm \ref{algo} shows the method by which we allocate these items to the agents in $\group_2$ and $\group_3$. 
Let $\mathcal{O}$ be a topological ordering of the agents with respect to the envy-graph. We order the agents in $\group_3$ according to $\mathcal{O}$ and ask them one by one to pick their most valuable remaining good. We then again ask agents in $\group_3$ to pick one more item according to the same topological ordering $\mathcal{O}$.  Afterwards, we order the agents in $\group_2$ according to $\mathcal{O}$ and ask them one by one to add the most desirable remaining item to their bundles. 

We now show that at the end of Step 2 the following conditions hold. The proof can be found in Appendix \ref{proofsec}.
\begin{mclaim}
\label{clm:secstep}
At the end of Step 2 the following conditions hold.
\begin{itemize}
	\item The allocation is $\EFR$ with respect to the agents in $\group_1$.	
	\item The allocation is $(3/4)$-$\EFR$ with respect to the agents in $\group_2$.
	\item The allocation is $(2/\varphi)$-$\EFR$ with respect to the agents in $\group_3$.

\end{itemize}
Since $2/\varphi < 3/4$, the allocation by the end of Step 2 is $(2/\varphi)$-$\efr$.
\end{mclaim}



\subsection{Step 3.}
In the third step, we use the envy-graph to allocate the remaining unallocated items.  We repeat the following steps until all the goods are allocated.
\begin{itemize}
	\item Find and eliminate all the directed cycles from the envy-graph. In order to eliminate all cycles in the envy-graph, we repeatedly find a directed cycle in the envy-graph. Let $i_1 \rightarrow i_2 \rightarrow \cdots \rightarrow i_k \rightarrow i_1$ be a cycle in envy-graph. By definition, each agent $i_j$ envies agent $i_{(j \Mod  k)+1}$, i.e., 
\begin{align*}
v_{i_j}(\allocation_{i_j}) < v_{i_j}\big(\allocation_{i_{(j \Mod  k)+1}}\big) \,,
\end{align*}
where $\allocation$ is the current allocation.
We then exchange the allocations of the agents that are in the cycle such that each agent $i_j$ receives $\allocation_{i_{(j \Mod  k)+1}}$. Note that this exchanging does not change bundles. Furthermore, the utility of each agent does not decrease. Hence, if the allocation is $\alpha$-$\efr$ before the exchange, it remains $\alpha$-$\efr$ after it (\textit{Lemma 6.1} in \cite{plaut2018almost}). Also, exchanging these allocations decreases the number of edges in the envy-graph. Thus, we eventually find an allocation such that its corresponding envy-graph is acyclic.
	\item Give an item to an agent that no-one envies. In the previous step we showed that we can always find an allocation such that its corresponding envy-graph is acyclic. Therefore, there should be a vertex in the envy-graph with no incoming edges. Let $i$ be the agent corresponding to this vertex. Since $i$ has no incoming edges in the envy-graph, no other agent envies $i$. At this step, we ask agent $i$ to pick his best item among all remaining goods.
\end{itemize}

The following Lemma shows the approximation guarantee of our algorithm. The proof can be found in Appendix \ref{sec:appen}.

\begin{restatable}{lemma}{stefr}
\label{lem:efxs3}
	Suppose that we are given a partial $\alpha$-$\efr$ allocation $\allocation$ such that for every agent $i$ and every remaining item $\ite$, we have $\valu_{i}(\ite) \le \alpha' \cdot \valu_{i}(\allocation_i)$ for some constant $\alpha' \le 1$. Then, the resulting allocation after performing the method mentioned above is $\min\{\alpha, \frac{1}{1+\alpha'}\}$-$\efr$.
\end{restatable}

We now show that at the beginning of Step 3, the valuation of every remaining item is small for all agents. The proof is available in Appendix \ref{sec:appen}.
\begin{restatable}{observation}{efrsth}
\label{before_step_3}
	Let $\allocation$ be the allocation after Step 2. Then for an agent $i$ and every remaining item $\ite$ we have
	\begin{itemize}
		\item If $i \in \group_1$, $\valu_i(\ite) \le  \valu_i(\allocation_i)/\varphi$.
		\item If $i \in \group_2$, $\valu_i(\ite) \le \valu_i(\allocation_i)/3$.
		\item If $i \in \group_3$, $\valu_i(\ite) \le  \valu_i(\allocation_i)/3$. 
	\end{itemize}
\end{restatable}

It follows from the observation above that for every agent $i$ the valuation of every remaining item is at most $\valu_i(\allocation_i)/\varphi$ after the second step of our algorithm. Recall that our allocation by the end of Step 2 is $(2/\varphi)$-$\efr$. Therefore, using Lemma \ref{lem:efxs3}, the allocation at the end of Step 3 is $\min\big\{\frac{2}{\varphi},\frac{1}{1+1/\varphi} \big\}$-$\efr$. Since $\varphi=\sqrt{3}+1$, we have
\begin{align*}
\frac{2}{\varphi} = \frac{1}{1+1/\varphi} = \sqrt{3}-1.
\end{align*} 
Therefore our final allocation is $\sqrt{3}-1 \approx 0.73$-$\efr$. This, coupled with the fact that all the steps can be implemented in polynomial time follows Theorem \ref{main2}.

\efrthm*
\section{Simple $(\phi-1)$-EFX Allocation} \label{sec:efx}

In this section, we show that our idea to use $\nsw$ matching as the first step of the allocation can easily give a $(\phi-1)$-$\efx$ allocation where $\phi = \frac{1+ \sqrt{5}}{2}$ is the golden ratio. The approximation ratio of our algorithm matches the state-of-the-art $(\phi-1)$ approximation result by Amanitidis et al. \cite{Amanatidis2018comparison}. Likewise our algorithm for $\efr$ allocation, our $(\phi-1)$-$\efx$ algorithm consists of 3 steps, namely $\nsw$ matching, allocation refinement, and envy-graph based allocation. The first and the third steps of our algorithm are almost the same as our previous algorithm. For the sake of completeness, we will briefly restate these steps in the rest of the section. The outline of our algorithm is represented in Algorithm \ref{algo:weakefx}.

\begin{algorithm}
\SetAlgoLined
\tcp{Step 1}
Allocate $\nsw$ matching\;
Let $r_i$ be envy-rank of an agent $i$. Divide the agents into groups $\group_1$ and $\group_2$ as follows. 	Agent $i$ belongs to $\group_1$ if $r_i > \phi$ and belongs to $\group_2$ otherwise\;
\tcp{Step 2}
Let $\mathcal{O}$ be a topological ordering of the agents with respect to the envy-graph\;
\ForEach{$i \in\group_2$ ordered by $\mathcal O$}
{
	Ask agent $i$ to pick his most valuable remaining item\;
}
\tcp{Step 3}
 \While{the allocation is not complete}{
 Eliminate all directed cycles in the envy-graph\;
 	Let $s$ be an arbitrary source in the envy-graph\;
	Ask agent $s$ to pick his most valuable remaining item\;
 }

\Return the allocation;
 \caption{The outline of the $(\phi-1)$-$\efx$ algorithm.}
 \label{algo:weakefx}
\end{algorithm}

\subsection{Step 1.}
In the first step, we allocate one item to each agent using $\nsw$ matching. Let $\allocation$ be the resulting allocation, and let $\phi=\frac{1+\sqrt{5}}{2}$ be the golden ratio. Based on the envy-rank of the agents, we divide them into $2$ groups $\group_1$ and $\group_2,$ as follows:
\begin{itemize}
	\item Agent $i$ belongs to $\group_1$ if $r_i > \phi$. 
	\item Agent $i$ belongs to $\group_2$ if $r_i\leq \phi$.	
\end{itemize}

\subsection{Step 2.}
In the second step, we allocate one item to each agent in $\group_2$ via the following process. let $\mathcal{O}$ be a topological ordering of the agents with respect to the envy-graph. We order the agents in $\group_2$ according to $\mathcal{O}$ and ask them one by one to pick their most valuable remaining good.

After this step, the bundle of every agent in $\group_2$ contains two items. For an agent $i \in \group_2$ we use $\{\ite_i, \ite'_i\}$ to denote the items allocated to this agent where $\ite'_i$ is the item allocated in Step $2$. We also use $\{\ite_i\}$ to denote the only item received by an agent $i \in \group_1$.  We show that at the end of Step 2 the following clamis hold.

\begin{restatable}{mclaim}{efxstf}
By the end of Step 2, the allocation is $\EFX$ with respect to the agents in $\group_1$.
\end{restatable}
\begin{proof}
Let $i$ be an agent in $\group_1$, we show that for every other agent $j$  we have 
$$\valu_i(\allocation_i) \geq \max_{\ite \in \allocation_j} \valu_i(\allocation_j \setminus \{\ite\}) \,,$$ 
so the allocation is $\efx$ from the agent $i$'s perspective. If $j \in \group_1$, then we have $|\allocation_j | = 1$. Therefore,
\begin{align*}
\max_{\ite \in \allocation_j} \valu_i(\allocation_j \setminus \{\ite\}) = 0 \,,
\end{align*}
and the claim clearly holds.

Consider an agent $ j \in \group_2$. At the end of Step 2, agent $j$ has two allocated items. By Observation \ref{obs:vremaining}, the valuation of the item $\ite'_j$ for agent $i$ is bounded by $\valu_i(\allocation_i)/r_i = \valu_i(\ite_i)/r_i$. Therefore,
\begin{align*}
\max_{\ite \in \allocation_j} \valu_i(\allocation_j \setminus \{\ite\}) &= \max \big\{\valu_i(\ite_j), \valu_i(\ite'_j) \big\} \\
&\le \max \big\{\valu_i(\ite_j), \valu_i(\ite_i)/r_i \big\} &\text{Observation \ref{obs:vremaining}.} \\
&\le \max \big\{\valu_i(\ite_i), \valu_i(\ite_i)/r_i \big\} &
\text{Observation \ref{clm:valloc}.} \\
&= \valu_i(\ite_i) = \valu_i(\allocation_i) \,. &\text{$r_i>1$.} 
\end{align*}
Therefore the allocation is $\efx$.
\end{proof} 

\begin{restatable}{mclaim}{efxsts}
By the end of Step 2, the allocation is $(\phi-1)$-$\efx$ with respect to the agents in $\group_2$.
\end{restatable}
\begin{proof}
Let $i$ be an agent in $\group_2$, we show that for every other agent $j$  we have 
$$\valu_i(\allocation_i) \geq (\phi-1) \cdot \max_{\ite \in \allocation_j} \valu_i(\allocation_j \setminus \{\ite\}) \,,$$ 
so the allocation is $(\phi-1)$-$\efx$ from agent $i$'s perspective. If $j \in \group_1$, then we have $|\allocation_j | = 1$, and the claim clearly holds.

Consider an agent $j$ in $\group_2$. We first consider the case that $\valu_i(\ite_i) < \valu_i(\ite_j)$. In this case, the position of agent $i$ in $\mathcal{O}$ is before agent $j$. Therefore, agent $i$ receives his second good before agent $j$, and we have 
\begin{align}
\label{eq:xweak21f}
\valu_i(\ite'_i) \ge \valu_i(\ite'_j) \,.
\end{align}
It follows that
\begin{align*}
\max_{\ite \in \allocation_j} \valu_i(\allocation_j \setminus \{\ite\}) &= \max \big\{\valu_i(\ite_j), \valu_i(\ite'_j) \big\} \\
&\le \max \big\{\valu_i(\ite_j), \valu_i(\ite'_i) \big\} &\text{By (\ref{eq:xweak21f}).} \\
&\le \max \big\{r_j \cdot \valu_i(\ite_i), \valu_i(\ite'_i) \big\} &
\text{Observation \ref{clm:valloc}.} \\
&\le \max \big\{\phi \cdot \valu_i(\ite_i), \valu_i(\ite'_i) \big\} &
\text{$r_j \le \phi$.} \\
&\le \phi \cdot \big(\valu_i(\ite_i) + \valu_i(\ite'_i)\big) \\
&= \phi \cdot \valu_i(\allocation_i) \,. 
\end{align*}
Therefore,
\begin{align*}
\valu_i(\allocation_i) \ge \frac{1}{\phi} \cdot  \max_{\ite \in \allocation_j} \valu_i(\allocation_j \setminus \{\ite\}) \,.
\end{align*}
Since $\frac{1}{\phi} = \phi -1$, it follows that our allocation is $(\phi-1)$-EFX. The other case is when $\valu_i(\ite_i) \ge \valu_i(\ite_j)$. In this case, we have
\begin{align*}
\max_{\ite \in \allocation_j} \valu_i(\allocation_j \setminus \{\ite\}) &= \max \big\{\valu_i(\ite_j), \valu_i(\ite'_j) \big\} \\
&\le \max \big\{\valu_i(\ite_i), \valu_i(\ite'_j) \big\} &\text{$\valu_i(\ite_i) \ge \valu_i(\ite_j)$.} \\
&\le \max \big\{\valu_i(\ite_i), \valu_i(\ite_i) \big\} &\text{Observation \ref{obs:vremaining}.} \\
&= \valu_i(\ite_i) \le \valu_i(\allocation_i) \,.
\end{align*}
Therefore, in this case the allocation is $\efx$ which completes the proof of the claim.
\end{proof} 

\subsection{Step 3.}
In the third step, we use the envy-graph to allocate the rest of the items. We show that at the beginning of Step 3, the valuation of every remaining item is small for all the agents.
\begin{restatable}{observation}{efxsth}
Let $\allocation$ be the allocation after Step 2. Then, for an agent $i$ and every remaining item $\ite$ we have
	\begin{itemize}
		\item If $i \in \group_1$, $\valu_i(\ite) \le  \valu_i(\allocation_i)/\phi$.
		\item If $i \in \group_2$, $\valu_i(\ite) \le \valu_i(\allocation_i)/2$. 
	\end{itemize}
\end{restatable}
\begin{proof}
Consider an agent $i \in \group_1$, then by Observation \ref{obs:vremaining} we have 
\begin{align*}
\valu_i(\ite) \le \valu_i(\allocation_i)/r_i \le \valu_i(\allocation_i)/\phi \,.
\end{align*}
Next, consider an agent $i \in \group_2$. This agent has two allocated items which are larger than every remaining item. Therefore, $\valu_i(\ite) \le \valu(\allocation_i)/2$.
\end{proof}

Therefore, for every agent $i$, the valuation of every remaining item is at most $\valu_i(\allocation_i)/\phi$. In order to complete the allocation we repeat the following steps until all the goods are allocated.
\begin{itemize}
	\item Find and eliminate all the directed cycles from the envy-graph.
	\item Allocate an item to an agent that no-one envies to him. 
\end{itemize}

\begin{lemma} [\cite{plaut2018almost}]
	Suppose that we are given a partial $\alpha$-$\efx$ allocation $\allocation$ such that for every agent $i$ and every remaining item $\ite$, we have $\valu_{i}(\ite) \le \alpha' \cdot (\allocation_i)$ for some constant $\alpha' \le 1$. Then, the resulting allocation after performing the method mentioned above is $\min\{\alpha, \frac{1}{1+\alpha'}\}$-$\efx$.
\end{lemma}

Recall that our allocation by the end of Step 2 is $(\phi-1)$-$\efx$. Therefore, by lemma above the approximation ratio of our approach is $\min\{ \phi -1, \frac{1}{1+1/\phi} \}$. Since $\frac{1}{1+1/\phi} = \phi -1$, it follows that our final allocation is $(\phi-1) \approx 0.61$-$\efx$.

%
%
%

\bibliography{ms}

\newpage
\appendix
\section{Missing proofs}\label{sec:appen}
\subsection{Missing proofs of Section \ref{prelim}}
\obsf*
\begin{proof}
		Assume $c = i_1 \rightarrow i_2 \rightarrow \ldots \rightarrow  i_k \rightarrow i_1$ 
 is an improving cycle. Then, it is easy to see that rotating  the goods over this cycle (i.e., reallocating $\allocation_{i_j}$ to agent $i_{j-1}$ for every $1< j \leq k$, and reallocating $\allocation_{i_1}$ to  agent $i_k$) yields a matching with a higher Nash social welfare.
\end{proof}
\valloc*
\begin{proof}
In the envy-ratio graph, the weight of the directed edge from $i$ to $j$ is $w_{i,j}= \frac{\valu_i(\allocation_j)}{\valu_i(\allocation_i)}$. Recall that $r_j$ is the maximum product of the weights of the edges in a path leading to $j$.  Since the edge from $i$ to $j$ is also a path leading to vertex $j$, we have $w_{i,j} \leq r_j$. Therefore, $\frac{\valu_i(\allocation_j)}{\valu_i(\allocation_i)} \le r_j$. Now consider a path $p$ leading to $i$ with the maximum product of the weights of the edges. Based on the definition of envy-rank, the product of the weights of the edges in $p$ is $r_i$. We can use the edge from $i$ to $j$ to extend  this path. This new path leads to $j$, and its product of the weights of the edges is $r_i \cdot w_{i,j}$. Therefore, we can say $r_j \geq r_i \cdot w_{i,j}$. Hence,
\begin{align*}
\frac{\valu_i(\allocation_j)}{\valu_i(\allocation_i)} = w_{i,j} \leq \frac{r_j}{r_i} \,.
\end{align*}
\end{proof}
\vremaining*
	\begin{proof}
First, for any agent $i$ and any remaining good $b$, we have $\valu_{i}(b) \leq \valu_i(\allocation_i)$, if not
$\nsw$ can be increased by giving $b$ to $i$ instead of $\allocation_i$. Moreover, consider a path $p=i_1 \rightarrow i_2 \rightarrow \ldots i_k \rightarrow i$ in the envy-ratio graph leading to $i$ with the maximum product of the weights of the edges. By the definition of the envy-rank the product of the weights of this path is $r_i$. By moving items along this path (giving $\allocation_i$ to $i_k$, $\allocation_{i_k}$ to $i_{k-1}$, etc.) and giving $b$ to agent $i$, the $\nsw$ will be multiplied by a factor of $r_i \cdot \frac{\valu_i(b)}{\valu_i(\allocation_i)}$. Since $\allocation$ is the allocation that maximizes $\nsw$, we have $r_i \cdot \frac{\valu_i(b)}{\valu_i(\allocation_i)} \leq 1$, and hence $\valu_i(b) \leq \frac{\valu_i(\allocation_i)}{r_i} $.
	\end{proof}
\subsection{Missing proofs of Section \ref{sec:efr}}
\poly*
\begin{proof}
Let $G=(U_1,U_2)$ be a bipartite graph that has a vertex for each agent in $U_1$ and has a vertex for every item in $U_2$. For every agent $i$ and every item $b$ we add an undirected edge with the weight of $\log \valu_i(\{b\})$ between their corresponding vertices. By finding a maximum weighted matching in this graph, we get an allocation $\allocation$ such that every agent has at most one allocated item. Also, this allocation maximizes $\sum_{i=1}^{n} \log \valu_i(\allocation_i)$. Therefore, this allocation also maximizes $\prod_{i=1}^{n} \valu_i(\allocation_i)$. Hence, $\allocation$ allocates at most one item to every agent and maximizes Nash social welfare.
\end{proof}
\stefr*
\begin{proof}
	The algorithm repeats the following steps until it allocates all items.
	\begin{itemize}
		\item Find and eliminate all the directed cycles from the envy-graph.
		\item Give an item to an agent that no-one envies.
	\end{itemize}
Consider the step in which the algorithm eliminates cycles. As we discussed earlier, this step does not change the approximation factor of the algorithm. Hence, if the allocation is $\alpha$-$\efr$ before this step, it remains $\alpha$-$\efr$ after it (See \textit{Lemma 6.1} in \cite{plaut2018almost} for more detail).

Consider the second step of the algorithm. In this step, the algorithm finds an agent such that no-one envies this agent, and it allocates an item to this agent. Suppose our algorithm allocates item $b$ to agent $i$. Since no-one envies agent $i$ before this step, for every other agent $j$, we have $\valu_j(\allocation_i) \leq \valu_j(\allocation_j)$ where $\allocation$ is the allocation of items before this step. In addition we have $\valu_j(b) \le \alpha' \cdot \valu_j(\allocation_j)$ since item $b$ was among unallocated items at the beginning of this step. Thus, we have
\begin{align*}
\valu_j(\allocation_i) + \valu_j(b) \le (1 + \alpha') \cdot  \valu_j(\allocation_j) \,.
\end{align*}
This means that after allocation $b$, no agent $j$ thinks the value of the bundle of agent $i$ is  $(1 + \alpha')$ times more than the valuation his bundle. Therefore, the allocation remains $\frac{1}{1+\alpha'}$-$\efr$.
 Hence, the final allocation is $\min\{\alpha, \frac{1}{1+\alpha'}\}$-$\efr$.
\end{proof}
\efrsth*
\begin{proof}
Consider an agent $i \in \group_1$, then by Observation \ref{obs:vremaining} we have 
\begin{align*}
\valu_i(\ite) \le \valu_i(\allocation_i)/r_i \le \valu_i(\allocation_i)/\varphi \,.
\end{align*}
Next consider an agent $i \in \group_2$. This agent has two allocated items. Let $\allocation_i = \{ \ite_i, \ite'_i \}$ be these items where $\ite_i$ is the item allocated using $\nsw$ matching. Since agent $i$ picks his best remaining item at Step 2 of our algorithm, we have
\begin{align}
\label{eq:s3g2-1}
\valu_i(\ite) \le \valu_i(\ite'_i) \,.
\end{align}
Since $\ite_i$ is allocated by $\nsw$ matching, by Observation \ref{obs:vremaining} we have
\begin{align}
\label{eq:s3g2-2}
\valu_i(\ite) &\le \valu_i(\ite_i)/r_i \nonumber \\ 
& \le \valu_i(\ite_i)/2 \,. &\text{Since $r_i>2$.}
\end{align}
It follows from (\ref{eq:s3g2-1}) and (\ref{eq:s3g2-2}) that
\begin{align*}
\valu_i(\ite) \le \big(\valu_i(\ite_i) + \valu_i(\ite'_i)\big)/3 = \valu_i(\allocation_i)/3 \,.
\end{align*}
The last case is when $i \in \group_3$. In this case agent $i$ has three allocated items which are all larger than every remaining item. Therefore $\valu_i(\ite) \le \valu(\allocation_i)/3$.
\end{proof}
\newpage
\section{Proof of Claim \ref{clm:secstep}} \label{proofsec}
After the Step 2 of the algorithm, the bundle of every agent in $\group_3$ contains three items. For an agent $i \in \group_3$ we use $\{\ite_i, \ite'_i, \ite''_i\}$ to denote the items allocated to this agent where $\ite'_i$ and $\ite''_i$ are the items allocated in Step $2$ and $\ite'_i$ has been allocated before $\ite''_i$. Also, the bundle of every agent in $\group_2$ contains two items. Similarly, for an agent $i \in \group_2$ we use $\{\ite_i,  \ite'_i\}$ to denote the allocated items of this agent where $\ite'_i$ is the item received in Step $2$. We also use $\{\ite_i\}$ to denote the only item received by an agent $i \in \group_1$.

We begin our analysis by showing the following claim.

\begin{mclaim}
\label{clm:exsum}
For an allocation $\allocation$ and agents $i$ and $j$, we have
\begin{align*}
\mathop{\mathlarger{\mathbb{E}}}_{\ite \sim \distribution_j} \big[\valu_i(\allocation_j \setminus \{\ite\})\big] = \frac{|\allocation_j|-1}{|\allocation_j|} \cdot \valu_i(\allocation_j) \,.
\end{align*}
\end{mclaim}
\begin{proof}
Distribution $\distribution_j$ selects each item in $\allocation_j$ with the probability of $1/|\allocation_j|$. Therefore,
\begin{align*}
\mathop{\mathlarger{\mathbb{E}}}_{\ite \sim \distribution_j} \big[\valu_i(\allocation_j \setminus \{\ite\})\big] &= \frac{1}{|\allocation_j|} \cdot \sum_{\ite \in \allocation_j} \valu_i(\allocation_j \setminus \{\ite\}) \\
&=\frac{1}{|\allocation_j|} \cdot \sum_{\ite \in \allocation_j} \sum_{\ite' \in \allocation_j \setminus \{\ite\}}  \valu_i(\{\ite'\}) \,. & \text{By Additivity assumption.}
\end{align*}
Each item in $\allocation_j$ appears $|\allocation_j|-1$ times in the above summation. Therefore,
\begin{align*}
\mathop{\mathlarger{\mathbb{E}}}_{\ite \sim \distribution_j} \big[\valu_i(\allocation_j \setminus \{\ite\})\big] &=\frac{1}{|\allocation_j|} \cdot \sum_{\ite \in \allocation_j} \sum_{\ite' \in \allocation_j \setminus \{\ite\}}  \valu_i(\{\ite'\}) \\
&= \frac{|\allocation_j|-1}{|\allocation_j|} \cdot \sum_{\ite \in \allocation_j} \valu_i(\{\ite\}) = \frac{|\allocation_j|-1}{|\allocation_j|} \cdot \valu_i(\allocation_j) \,.
\end{align*}
\end{proof}

First we show that at the end of Step 2, the allocation is $\efr$ for the agents in $\group_1$.
\begin{mclaim}
By the end of Step 2, the allocation is $\EFR$ with respect to the agents in $\group_1$.
\end{mclaim}

\begin{proof}
Let $i$ be an agent in $\group_1$, we show that for every other agent $j$  we have 
$$\E_{\ite \sim \distribution_j} \big[\valu_i(\allocation_j \setminus \{\ite\})\big]\le \valu_i(\allocation_i) ,$$
 so the allocation is $\EFR$ from the agent $i$'s perspective.
\begin{itemize}
\item If $j \in \group_1$, then we have $|\allocation_j | = 1$ and the claim clearly holds.
\item If $j \in \group_2$, then at the end of Step 2, agent $j$ has two allocated items. By Observation \ref{obs:vremaining} the valuation of the item $\ite'_j$ for agent $i$ is bounded by $\valu_i(\allocation_i)/r_i = \valu_i(\ite_i)/r_i$. Therefore,
\begin{align*}
\E_{\ite \sim \distribution_j} \big[\valu_i(\allocation_j \setminus \{\ite\})\big] &= \frac{\valu_i(\ite_j) + \valu_i(\ite'_j)}{2} & \text{Claim \ref{clm:exsum}.}
\\
&\le \frac{\valu_i(\ite_j) + \valu_i(\ite_i)/r_i}{2} & \text{ Observation \ref{obs:vremaining}.} \\
&\le \frac{\valu_i(\ite_i) + \valu_i(\ite_i)/r_i}{2} & \text{ Observation \ref{clm:valloc}.} \\
&=  \frac{1+1/r_i}{2} \cdot \valu_i(\ite_i)\,.
\end{align*}
Since agent $i$ is in $\group_1$, we have $r_i \ge \varphi$. It follows that
\begin{align*}
\E_{\ite \sim \distribution_j} \big[\valu_i(\allocation_j \setminus \{\ite\})\big] &\le \frac{1+1/r_i}{2} \cdot \valu_i(\ite_i) \\
&\le   \frac{1+1/\varphi}{2} \cdot \valu_i(\ite_i) \\
&= \frac{\varphi+1}{2\varphi} \cdot \valu_i(\ite_i) \,.
\end{align*}
Since $\varphi=\sqrt{3}+1$, we have $\frac{\varphi+1}{2\varphi} <1$. Therefore,
\begin{align*}
\E_{\ite \sim \distribution_j} \big[\valu_i(\allocation_j \setminus \{\ite\})\big] &\le \frac{\varphi+1}{2\varphi} \cdot \valu_i(\ite_i)  <  \valu_i(\ite_i) = \valu_i(\allocation_i) \,.
\end{align*}
Therefore, in this case the allocation is $\efr$.
\item The only remaining case is when agent $j$ is in $\group_3$. By Observation \ref{obs:vremaining}, valuation of $\ite'_j$ and $\ite''_j$ for agent $i$ is at most $\valu_i(\ite_i)/r_i$. Also, by Observation \ref{clm:valloc}, valuation of $\ite_j$ for agent $i$ is at most $r_j \cdot \valu_i(\ite_i)/r_i$. Therefore,
\begin{align*}
\E_{\ite \sim \distribution_j} \big[\valu_i(\allocation_j \setminus \{\ite\})\big] &= \frac{2}{3} \cdot \big(\valu_i(\ite_j) + \valu_i(\ite'_j) + \valu_i(\ite''_j) \big) & \text{Claim \ref{clm:exsum}.}
\\
&\le \frac{2}{3} \cdot \big(\valu_i(\ite_j) + 2\valu_i(\ite_i)/r_i\big) & \text{Observation \ref{obs:vremaining}.} \\
&\le \frac{2}{3} \cdot \big(r_j \cdot \valu_i(\ite_i)/r_i + 2\valu_i(\ite_i)/r_i\big) & \text{Observation \ref{clm:valloc}.} \\
&=  \frac{2 r_j/ r_i +4/r_i}{3} \cdot \valu_i(\ite_i)  \,.
\end{align*}
Recall that agents $i$ and $j$ are in $\group_1$ and $\group_3$ respectively. Therefore, $r_i \ge \varphi$ and $r_j \le 2$. We then have
\begin{align*}
\E_{\ite \sim \distribution_j} \big[\valu_i(\allocation_j \setminus \{\ite\})\big] &\le \frac{2 r_j/ r_i +4/r_i}{3} \cdot \valu_i(\ite_i)\\
  &\le \frac{8/\varphi}{3} \cdot \valu_i(\ite_i)
  \end{align*}
  Since $\varphi=\sqrt{3}+1$, we have $\frac{8/\varphi}{3} <1$. Therefore,
\begin{align*}
\E_{\ite \sim \distribution_j} \big[\valu_i(\allocation_j \setminus \{\ite\})\big] \le  \frac{8/\varphi}{3} \cdot \valu_i(\ite_i) < \valu_i(\ite_i) = \valu_i(\allocation_i) \,.
\end{align*}
Therefore the allocation is $\EFR$.
\end{itemize}
\end{proof} 

Now we show that at the end of Step 2, the allocation is $(3/4)$-$\EFR$ with respect to the agents in $\group_2$.

\begin{mclaim}
By the end of Step 2, the allocation is $(3/4)$-$\EFR$ with respect to the agents in $\group_2$.
\end{mclaim}
\begin{proof}
Let $i$ be an agent in $\group_2$, we show that for every other agent $j$, we have 
$$\E_{\ite \sim \distribution_j} \big[\valu_i(\allocation_j \setminus \{\ite\})\big] \le 4/3 \cdot \valu_i(\allocation_i) \,,$$
therefore the allocation is $(3/4)$-$\efr$. 
\begin{itemize}
\item For an agent $j$ in $\group_1$, only one item is allocated to this agent and the claim clearly holds.

\item Consider an agent $j$ in $\group_2$. Recall that $\allocation_j=\{\ite_j, \ite'_j\}$ is the bundle of this agent. We first consider the case that $\valu_i(\ite_i) < \valu_i(\ite_j)$. In this case the position of agent $i$ in the topological order $\mathcal{O}$ is before agent $j$. Therefore, agent $i$ receives his second good before agent $j$, and we have 
\begin{align}
\label{eq:g22s}
\valu_i(\ite'_i) \ge \valu_i(\ite'_j) \,.
\end{align}
Moreover, by Observation \ref{clm:valloc} we have
\begin{align}
\label{eq:g22m}
\frac{\valu_i(\ite_j)}{\valu_i(\ite_i)} \le \frac{r_j}{r_i} < \frac{\varphi}{2} \,.
\end{align} 
It follows that
\begin{align*}
\E_{\ite \sim \distribution_j} \big[\valu_i(\allocation_j \setminus \{\ite\})\big] &=
\frac{\valu_i(\ite_j) + \valu_i(\ite'_j)}{2} & \text{Claim \ref{clm:exsum}.} \\
&\le \frac{\valu_i(\ite_j) + \valu_i(\ite'_i)}{2}  & \text{By (\ref{eq:g22s}).} \\
&< \frac{\varphi/2 \cdot \valu_i(\ite_i) + \valu_i(\ite'_i)}{2} & \text{ By (\ref{eq:g22m}).} \\
&\le \frac{\varphi}{4} \cdot \big(\valu_i(\ite_i)+\valu_i(\ite'_i)\big) & \text{ Since $\varphi/2>1$.} \\
\end{align*}
Since $\varphi=\sqrt{3}+1$, we have $\frac{\varphi}{4} <1$. Therefore,
\begin{align*}
\E_{\ite \sim \distribution_j} \big[\valu_i(\allocation_j \setminus \{\ite\})\big] & \le \frac{\varphi}{4} \cdot \big(\valu_i(\ite_i)+\valu_i(\ite'_i)\big) \\
&< \valu_i(\ite_i)+\valu_i(\ite'_i) = \valu_i(\allocation_i) \,.
\end{align*}
Therefore, if $\valu_i(\ite_i) < \valu_i(\ite_j)$, the allocation is $\EFR$. The other case is when $\valu_i(\ite_i) \ge \valu_i(\ite_j)$. In that case we have
\begin{align*}
\E_{\ite \sim \distribution_j} \big[\valu_i(\allocation_j \setminus \{\ite\})\big] &=
\frac{\valu_i(\ite_j) + \valu_i(\ite'_j)}{2} & \text{Claim \ref{clm:exsum}.} \\
&\le \frac{\valu_i(\ite_j) + \valu_i(\ite_i)/r_i}{2} &\text{ Observation \ref{obs:vremaining}.} \\
&\le  \frac{\valu_i(\ite_i) + \valu_i(\ite_i)/r_i}{2} &\text{ $\valu_i(\ite_i) \ge \valu_i(\ite_j)$.} \\
&= \frac{1+1/r_i}{2} \cdot \valu_i(\ite_i) \\
&< \frac{3}{4} \cdot \valu_i(\ite_i) \le \frac{3}{4} \cdot \valu_i(\allocation_i)  \,, & \text{$r_i>2$.}  
\end{align*}
therefore the allocation in this case is $\EFR$.

\item The only remaining case is when agent $j$ is in $\group_3$.  By Observation \ref{obs:vremaining}, valuation of $\ite'_j$ and $\ite''_j$ for agent $i$ is at most $\valu_i(\ite_i)/r_i$. Also, by Observation \ref{clm:valloc}, valuation of $\ite_j$ for agent $i$ is at most $r_j \cdot \valu_i(\ite_i)/r_i$. Thus,
\begin{align*}
\E_{\ite \sim \distribution_j} \big[\valu_i(\allocation_j \setminus \{\ite\})\big] &= \frac{2}{3} \cdot \big(\valu_i(\ite_j) + \valu_i(\ite'_j) + \valu_i(\ite''_j) \big) & \text{Claim \ref{clm:exsum}.}
\\
&\le \frac{2}{3} \cdot \big(\valu_i(\ite_j) + 2\valu_i(\ite_i)/r_i\big) & \text{Observation \ref{obs:vremaining}.} \\
&\le \frac{2}{3} \cdot \big(r_j \cdot \valu_i(\ite_i)/r_i + 2\valu_i(\ite_i)/r_i\big) & \text{Observation \ref{clm:valloc}.} \\
&\le \frac{2}{3} \cdot \big(\valu_i(\ite_i) + 2\valu_i(\ite_i)/r_i\big) & \text{$r_j <r_i$.} \\
&= \frac{2+4/r_i}{3} \cdot \valu_i(\ite_i)  \\
&<  \frac{4}{3} \cdot \valu_i(\ite_i) \le  \frac{4}{3} \cdot \valu_i(\allocation_i) \,. & \text{$r_i > 2$.} \\
\end{align*}
Therefore the allocation is $(3/4)$-$\EFR$.
\end{itemize}
\end{proof}

Now we show that by the end of Step 2, the allocation is $(2/\varphi)$-$\EFR$ with respect to the agents in $\group_3$.

\begin{mclaim}
By the end of Step 2, the allocation is $(2/\varphi)$-$\EFR$ with respect to the agents in $\group_3$.
\end{mclaim}
\begin{proof}
Let $i$ be an agent in $\group_3$, we show that for every other agent $j$, we have 
$$\E_{\ite \sim \distribution_j} \big[\valu_i(\allocation_j \setminus \{\ite\})\big] \le \varphi/2 \cdot\valu_i(\allocation_i) \,,$$

therefore the allocation is $(2/\varphi)$-$\efr$.
\begin{itemize}
\item For an agent $j$ in $\group_1$, the claim clearly holds since this agent has only one allocated item.

\item Consider an agent $j \in \group_2$. Agent $i$ receives his second item before agent $j$ in Step 2 of our algorithm. Thus, we have
\begin{align}
\label{eq:g32s}
\valu_i(\ite'_i) \ge \valu_i(\ite'_j) \,.
\end{align} 
Also, by Observation \ref{clm:valloc}, valuation of $\ite_j$ for agent $i$ is at most $r_j \cdot \valu_i(\ite_i)$. Thus,
\begin{align*}
\E_{\ite \sim \distribution_j} \big[\valu_i(\allocation_j \setminus \{\ite\})\big] &=
\frac{\valu_i(\ite_j) + \valu_i(\ite'_j)}{2} & \text{Claim \ref{clm:exsum}.} \\
&\le \frac{\valu_i(\ite_j) + \valu_i(\ite'_i)}{2} & \text{By (\ref{eq:g32s}).}\\
&\le \frac{r_j \valu_i(\ite_i) + \valu_i(\ite'_i)}{2} & \text{Observation \ref{clm:valloc}.} \\
&\le \frac{\varphi \valu_i(\ite_i) + \valu_i(\ite'_i)}{2} & \text{ $r_j \le \varphi$.} \\
&\le \frac{\varphi}{2} \cdot \big(\valu_i(\ite_i)+\valu_i(\ite'_i)\big) & \text{Since $\varphi > 1$.} \\
&\le \frac{\varphi}{2} \cdot \valu_i(\allocation_i) \,.
\end{align*}
Therefore the allocation is $(2/\varphi)$-$\EFR$.

\item The remaining case is when agent $j$ is in $\group_3$. Consider the case that $\valu_i(\ite_i) < \valu_i(\ite_j)$. In this case the position of agent $i$ in the topological order $\mathcal{O}$ is before agent $j$. Therefore, in Step 2 of our algorithm, agent $i$ receives his second and third items before agent $j$, and we have the followings.
\begin{align}
\label{eq:g33s1}
\valu_i(\ite'_i) \ge \valu_i(\ite'_j) \,, 
\end{align}
and
\begin{align}
\label{eq:g33s2}
\valu_i(\ite''_i) \ge \valu_i(\ite''_j) \,.
\end{align}
Also, by Observation \ref{clm:valloc}, valuation of $\ite_j$ for agent $i$ is at most $r_j \cdot \valu_i(\ite_i)$. Thus,
\begin{align*}
\E_{\ite \sim \distribution_j} \big[\valu_i(\allocation_j \setminus \{\ite\})\big] &= \frac{2}{3} \cdot \big(\valu_i(\ite_j) + \valu_i(\ite'_j) + \valu_i(\ite''_j) \big) & \text{Claim \ref{clm:exsum}.}
\\
&\le  \frac{2}{3} \cdot \big(\valu_i(\ite_j) + \valu_i(\ite'_i) + \valu_i(\ite''_i) \big) & \text{By (\ref{eq:g33s1}) and (\ref{eq:g33s2}).} \\
&\le  \frac{2}{3} \cdot \big(r_j \valu_i(\ite_i) + \valu_i(\ite'_i) + \valu_i(\ite''_i) \big) & \text{Observation \ref{clm:valloc}.} \\
&\le \frac{2}{3} \cdot \big(2 \valu_i(\ite_i) + \valu_i(\ite'_i) + \valu_i(\ite''_i) \big) & \text{ $r_j \le 2$.} \\
& \le \frac{4}{3} \cdot \big(\valu_i(\ite_i) + \valu_i(\ite'_i) + \valu_i(\ite''_i) \big) \\
& = \frac{4}{3} \cdot \valu_i(\allocation_i) \,.
\end{align*}
Thus, the allocation is $(3/4)$-$\EFR$. Since $2/\varphi < 3/4$, the allocation is also $(2/\varphi)$-$\EFR$.

The other case is when $\valu_i(\ite_i) \ge \valu_i(\ite_j)$. In this case  agent $i$ receives $\ite'_i$ prior to when agent $j$ receives $\ite''_j$, and we have
\begin{align}
\label{eq:g33m}
\valu_i(\ite'_i) \ge \valu_i(\ite''_j)
\end{align}
Also, by Observation \ref{obs:vremaining}, valuation of $\ite'_j$ for agent $i$ is at most $\valu_i(\ite_i)$. Thus,
\begin{align*}
\E_{\ite \sim \distribution_j} \big[\valu_i(\allocation_j \setminus \{\ite\})\big] &= \frac{2}{3} \cdot \big(\valu_i(\ite_j) + \valu_i(\ite'_j) + \valu_i(\ite''_j) \big) & \text{Claim \ref{clm:exsum}.}
\\
&\le \frac{2}{3} \cdot \big(\valu_i(\ite_i) + \valu_i(\ite'_j) + \valu_i(\ite''_j) \big) & \text{$\valu_i(\ite_i) \ge \valu_i(\ite_j)$.} \\
&\le \frac{2}{3} \cdot \big(\valu_i(\ite_i) + \valu_i(\ite'_j) + \valu_i(\ite'_i) \big) & \text{By (\ref{eq:g33m}).}\\
&\le \frac{2}{3} \cdot \big(\valu_i(\ite_i) + \valu_i(\ite_i) + \valu_i(\ite'_i) \big) & \text{Observation \ref{obs:vremaining}.} \\
&\le \frac{4}{3} \cdot \big(\valu_i(\ite_i) + \valu_i(\ite'_i) \big)\\
&\le \frac{4}{3} \cdot \valu_i(\allocation_i) \,,
\end{align*}
and the allocation is $(3/4)$-$\EFR$ as well as  $(2/\varphi)$-$\EFR$. 
\end{itemize}
\end{proof}
\end{document}